\documentclass[english,utf8]{article}

\usepackage{hyperref}
\usepackage{latexsym,amssymb}
\usepackage{amsmath}
\usepackage{wasysym,stmaryrd}
\usepackage{color,graphicx}
\usepackage{float} 
\usepackage[english]{babel}

\usepackage{wrapfig}
\usepackage[all]{xy}
\input{xy} \xyoption{all}
\newtheorem{definition}{Definition}[section]
\newtheorem{lemma}[definition]{Lemma}
\newtheorem{theorem}[definition]{Theorem}
\newtheorem{corollary}[definition]{Corollary}

\newcounter{excount}

\newenvironment{proof}[1]{\par \textit{Proof:} \\ #1}{\par}
\newenvironment{remark}[1]{\par \textit{Remark:} \\ #1}{\par}
\newenvironment{concept}[2]{\par \textit{Main Concept: #1} \\ #2}{\par}
\newenvironment{example}[2]{\par \textit{Example \stepcounter{excount} \arabic{excount}: #1} \\ #2}{\par}
\newenvironment{keywords}[1]{\par\parindent0mm \textbf{Keywords:} #1}{\par}
\newcommand{\komtar}[2]{}  

\makeatletter
\newcommand{\xRightarrow}[2][]{\ext@arrow 0359\Rightarrowfill@{#1}{#2}}
\makeatother

\newcommand{\reconnet}{\textsc{ReCon}\textsc{Net}}
\newcommand{\HRPN}{hierarchical reconfigurable Petri net}

\newcommand{\deriv}[1]{\xRightarrow{#1}}

\newcommand{\nach}[1]{\stackrel{#1}{\to}}
\newcommand{\von}[1]{\stackrel{#1}{\gets}}

\newcommand{\Nat}{{\mathbb{N}}}
\newcommand{\M}{\ensuremath{\mathcal{M}} }
\newcommand{\E}{\ensuremath{\mathcal{E}} }
\newcommand{\fire}[1]{\ensuremath{[#1>}}
\newcommand{\R}{\ensuremath{\mathcal{R}} }
\newcommand{\g}{\ensuremath{\mathfrak{g}} }

\newcommand{\pl}{pl}
\newcommand{\tl}{tl}
\newcommand{\mrk}{M}

\newcommand{\categ}[1]{\ensuremath{\mathbf{ #1} } }
\newcommand{\cC}{\categ{C}}
\newcommand{\cSets}{\categ{Sets}}
\newcommand{\lSets}{\categ{lSets}}
\newcommand{\poSets}{\categ{poSets}}
\newcommand{\poSetsg}{\categ{poSetsg}}

\newcommand{\cPT}{\categ{PT}}
\newcommand{\cAHL}{\categ{AHL}}
\newcommand{\cdecoPT}{\categ{decoPT}}
\newcommand{\cPTs}{\categ{PTs}}
\newcommand{\cAHLs}{\categ{AHLs}}
\newcommand{\cdecoPTs}{\categ{decoPTs}}

\newcommand{\A}[2]{A_{#1}^{#2}}

\newcommand{\xyInc}[2]{\ar@{^{(}->}[#1]^{#2} }

\begin{document}
\title{Subtyping for Hierarchical, Reconfigurable Petri Nets}

\author{Julia Padberg\footnote{%
   Hamburg University of Applied Sciences, Hamburg,Germany,
	email{julia.padberg@haw-hamburg.de} }}

\maketitle
\begin{abstract}
Hierarchical Petri nets allow  a more abstract view  and
reconfigurable Petri nets  model dynamic structural adaptation. In this contribution we present   the combination of reconfigurable Petri nets and hierarchical Petri nets yielding   hierarchical structure  for  reconfigurable Petri nets.  Hierarchies are established by substituting transitions by subnets. These subnets are themselves reconfigurable, so they are supplied with their own set of rules. Moreover, global rules that can be applied in all of the net, are provided.
\begin{keywords}
$\M$-adhesive transformation systems, subtyping, Reconfigurable Hierarchical Petri Nets
\end{keywords}

\end{abstract}

\section{Introduction}
\label{s.intro}

Modelling modern systems comes with a lot of different challenges some of which can be eased
by the use of appropriate models. 
A well known technique for system modeling is the usage of Petri nets. Petri nets provide a graphical language for constructing system models. These models can be used for simulations and to analyze the model's properties. This allows locating possible faults in the system at earlier stages which also can decrease overall development costs.The increasing sizes of modern systems result in models becoming rather large and possibly hard to comprehend. The addition of an abstraction layer can counteract this issue. Hierarchical Petri nets use hierarchy to break down the complexity of a large model,
by dividing it into a number of submodels. This helps to concentrate on a specific system part without the need to oversee the whole system. Also submodels can be reused with little afford at multiple location in the same system or even in a different system where similar components are needed.
The analysis and verification of a hierarchical Petri net requires more effort than it's counter part with no hierarchy. 

Advanced systems that need dynamic structural adaptation can be modelled using reconfigurable Petri nets, an approach for dynamic changes in Petri nets. A reconfigurable Petri net  consist of a Petri net and set of rules that modify the Petri net's structure at runtime. They improve the expressiveness of  Petri nets as they increase   flexibility and  change while allowing the transitions to fire.

Reconfigurable Petri nets have been used in many different  application areas, that require both the representation of  their processes and of the system changes within one model. examples are  concurrent systems \cite{LO04}, mobile ad-hoc networks \cite{PHE07}, workflows in dynamic infrastructures \cite{HEP08},
 communication spaces \cite{MGH10,Gab14}, ubiquitous computing \cite{GNH12,Hoffmann10},  flexible manufacturing systems \cite{RPN_FMS_12}, reconfigurable manufacturing systems \cite{kahloul2016designing}).
In \cite{PK17} a comprehensive overview of reconfigurable Petri nets is given, including theoretical foundations and application areas.

\reconnet ~ \cite{reconnet,PEOH12} is a modelling and simulation tool that allows the design of reconfigurable Petri nets. 

\HRPN s combine the hierarchical Petri net type and reconfigurable Petri net type into one, allowing a focused design of submodels and their reusability and the ability for dynamic changes at runtime. 

\komtar{Inhalte }{JP}

\begin{concept}{Hierarchy as a syntactic extension}
The hierarchy based on  transitions being replaced by  subnets is given as a syntactic abbreviation. The hierarchical reconfigurable net is
defined purely by it's flattening into a  reconfigurable net.
\end{concept}

This fundamental  design decision has the following advantages:
First, only the consistency of the flattening construction can be guaranteed using well-known results, but no further semantic correctness needs to be proven. Second, this corresponds directly to the intended implementation of hierarchies in \reconnet. And last,  the transformation systems needs not to be shown for another category of hierarchical nets.

This paper is the detailed version of \cite{LP18} with emphasis on the subtyping and is organized as follows: 
We start by given a very simple introductory example of \HRPN s in Sect.~\ref{s.intEx}.
The subsequent section introduces reconfigurable place/transition nets with labels.  In Sect.~\ref{s.relwork}
 related work, namely approaches to hierarchical Petri nets and hierarchical graph transformations are discussed.
Sect. \ref{s.hier} elaborates the formal definition of \HRPN s and their flattening, the definition of transformation rules  and proves the correctness of the flattening construction.
The integration of reconfigurable hierarchical Petri nets into the  simulation tool for reconfigurable Petri nets \reconnet is discussed in Sect.~\ref{s.hier_reconnet}.   The conclusion in Sect.~\ref{s.conc} completes this paper.

\section{Subtyping of Labels in $\M$-Adhesive Transformation Systems}
\label{s.subtype}

\subsection{Labels and Sub-labels}
\label{ss.labels}


	\begin{definition}{Category of labelled sets $\lSets$}
	Given a partial order $(A,\le, \g)$ with a greatest $a\le \g$ for all $a \in A$ as the label alphabet. 
	
		The category of labelled sets with sub-labelling $lSets$ over
		label alphabet $(A,\le, \g)$ has  $(S,l:S \to A)$ as objects and order-preserving maps
	$f: (S,l) \to (S',l')$  so that $l'\circ f (x) \le l(x)$ for all $x\in S$ as morphisms.
\end{definition}

\begin{lemma}{Adjunction between  \cSets \ and \lSets}\label{l.adjunction}
The left adjoint  functor $F:\cSets \to \lSets$ is given by $F(S_1\nach{f}S_2)= (S_1,l_1) \nach{f} (S_2,l_2) $ 
where $l_i: S_i \to (A,\le)$ so that $l_i(x) = \g$ for $i=1,2$ yields the greatest element of $A$.
The right adjoint functor $G: \lSets \to \cSets$ is defined by $G( S_1,l_1) \nach{g} (S_2,l_2) = S_1 \nach{g} S_2$.

The counit is the natural transformation $\epsilon : F\circ G \to id_\lSets$ with $\epsilon_S=id_S$ an order-preserving map
since for any $s\in S$ we have $l\circ id_S(s) =l(s) \le l_{g}(s)$ .
The  unit is the natural transformation $\eta : id_\cSets \to G\circ F$ with   $\eta_S=id_S$.
\end{lemma}
\begin{proof}
Let be $M\in Obf(\cSets)$ and $(S,l)\in Obj(\lSets)$:
$$\xymatrix{
              G(S,l) =S  \ar[rr]^{\eta_{G(S,l)} = id_S} \ar[drr]|{id_{G(S,l)}=id_S}
         && GFG(S,l) = S \ar[d]^{G(\epsilon_S)}    \\
         && G(S,l)=S
}
\hfill
\xymatrix{
        F(M) =(M,l_{g}) \ar[rr]^{F(\eta_M) = id_M}  \ar[drr]|{id_{F(M)}=id_M}
	&& FGF(M)=	(M,l_{g})  \ar[d]^{\epsilon_{F(M)} = id_{F(M)}}\\
	&& F(M)=	(M,l_{g})
}
$$

Obviously, the composition of identities leads to the corresponding identity with  
$\epsilon_{F(M)} \circ F(\eta_M) =  id_{F(M)} \circ F(id_M) = id_{F(M)}$ and $G(\epsilon_{(S,l)}) \circ \eta_{G(S,l)}= id_{G(S,l)} \circ  id_{G(S,l)}  =id_{G(S,l)}$.
\end{proof}
So, we know that $F$ preserves  colimits ans $G$ preserves limits. Next we show   initial objects and pushouts of $\M$-morphisms. 

\begin{definition}{Class $\M$}
\label{d.M}
The class $\M$ is given by the class of strict order preserving, injective mappings, i.e
$f:(S,l) \to (S',l)$ so that $f$ is injective and  $l = l'\circ f$.
\end{definition}
\begin{lemma}{Initial Object and $\M$-Pushouts in \lSets}
\label{l.POlSets}
The initial object is $(\emptyset,\emptyset)$.\\ 
Given the span $(S_1,l_1) \von{f} (S_0,l_0) \nach{g} (S_2,l_2)$ with $f\in \M$, then there exists  the pushout
	     $(S_1,l_1) \nach{g'} (S_3,l_3) \von{f'} (S_2,l_2)$. Moreover \M \ is stable under pushouts.
\end{lemma}
\begin{proof}
\begin{enumerate}
	\item The initial object is $(\emptyset,\emptyset)$ as there is the empty 
	mapping to each labelled  set in \lSets \ and it is order-preserving.
		\item For $(S_1,l_1) \von{f} (S_0,l_0) \nach{g} (S_2,l_2)$ with $f\in \M$
	there is in $\cSets$ the span
	       $S_1\von{f} S_0 \nach{g} S_2$ and its pushout $S_1 \nach{\overline{g}} \overline{S_3} \von{\overline{f}} S_2$, see pushout $(PO1)$ in Fig.~ \ref{f.lSetPO}				
				with 
				\begin{equation*}
	         \label{eq.l3}
				  l_3(s) =\begin{cases}
					             l_1(s_1) & \text{ if } s=g'(s_1) \text{  and } s \notin f'(S_2)\\
											 l_2(s_2) & \text{ if } s=f'(s_2) 
					             \end{cases}
         \end{equation*}
				The morphism $f':(S_2, l_2) \to S_3,l_3)$ is obviously well defined  and $f' \in \M$ and  $g':(S_1, l_1) \to (S_3,l_3)$ is well defined
				since :\\
				for $s_1\notin  f(S_0)$  we have 	$l_1(s_1)= l_3(g'(s_1))$ and \\
				for $s_1=f(s_0)$  we have \\
				$l_1(s_1) = l_1\circ f(s_0) = l_0(s_0) \ge l_2 \circ g (s_0)  = l_3 \circ f' \circ g (s_0) = l_3 \circ g' \circ f (s_0) = l_3\circ g'(s_1)$\\
    $\M$ is stable under pushouts, since  $f \in \M$.
			 \begin{figure}[H]
					$\xymatrix{
					     (PO1) \text{ in } \cSets: \\
									S_0 \ar[rr]^{f \in \M} \ar[dd]_{g}  \ar[dr]|{l_0}
						 && S_1  \ar[dd]^{g'}    \ar[dl]|{l_1}  \ar@/^4mm/[dddrr]^{g''} \\
						   & (A,\le) \\
							     S_2  \ar[rr]_{f'}  \ar[ur]|{l_2}  \ar@/_4mm/[drrrr]^{f''}
						&&  S_3 \ar[ul]|{l_3} \ar[drr]^h\\
				&&&&  S
					}
					$\hfill
					$\xymatrix{
					      \lSets: \\
									(S_0,l_0)  \ar[rr]^{f \in \M} \ar[dd]_{g}  \ar@{}[ddrr]|{(PO2)}
						 && (S_1,l_1)  \ar[dd]^{g'} \ar@/^4mm/[dddrr]^{g''}   \\~\\
							     (S_2,l_2)  \ar[rr]_{f'}  \ar@/_4mm/[drrrr]^{f''}
						&&  (S_3,l_3)  \ar[drr]^h\\
				&&&&  (S,l)
						}
						$
				 \caption{
	         \label{f.lSetPO} Pushout Construction in $\lSets$}
			 \end{figure}
			
			$(PO2)$ commutes in $\lSets$ and given a labelled set $(S,l)$, so that $g'' \circ f = f'' \circ g$, then there is in $\cSets$ the 
			unique induced morphism $h: S_3\to S$ so that $h\circ g' = g''$ and $h\circ f'= f''$.\\
			$h$ is well defined in $\lSets$ as well, since :\\
			for $s_3= g'(s_1)$ and $s_3 \notin f'(S_2)$ we have\\
			$l_3(s_3) = l_3\circ g'(s_1) = l_1 (s_1) \ge l \circ g''(s_1) = l \circ h \circ g'(s_1) = l \circ h(s_3)$
			for $s_3= f'(s_2)$  we have\\
			$l_3(s_3) = l_3\circ f'(s_2) = l_2 (s_2) \ge l \circ f''(s_2) = l \circ h \circ f'(s_2) = l \circ h(s_3)$\\
			So, $(2)$ is pushout in $\lSets$.
	
			\end{enumerate}
\end{proof}

\begin{lemma}{$\lSets$ has pullbacks}
\label{l.PBlSets}
Given the co-span $(S_1,l_1) \nach{g} (S_0,l_0) \von{f} (S_2,l_2)$, then there exists  the pullback 
$(S_1,l_1) \von{f'} (S_3,l_3) \nach{g'} (S_2,l_2)$ in the category $\lSets$.  Moreover \M \ is stable under pullbacks.
\end{lemma}
\begin{proof}
Given the co-span  $(S_1,l_1) \nach{g} (S_0,l_0) \von{f} (S_2,l_2)$, then there is the pullback $(PB1)$ in $\cSets$ with $S_3 \cong \{ (s_1,s_2) \mid g(s_1) = f(s_2)\}$, the projections $g'$ and $f'$, and  the induced morphisms $h$, so that 
$g' \circ h = g''$ and $f' \circ h= f''$.

			 \begin{figure}[H]
					$\xymatrix{
					      \cSets: \\
								S \ar@/^4mm/[drrrr]^{f''} \ar@/_4mm/[dddrr]^{g''}  \ar[drr]^h\\
						 && S_3  \ar[rr]^{f'} \ar[dd]_{g'}  \ar@{}[ddrr]|{(PB1)}
						 && S_1  \ar[dd]^{g}   \\~\\
						 && S_2  \ar[rr]_{f}  
						&&  S_0 
						}
						$
						$\xymatrix{
					      \lSets: \\
								  (S,l) \ar@/^4mm/[drrrr]^{f''} \ar@/_4mm/[dddrr]^{g''}  \ar[drr]^h \\
						 &&  (S_3,l_3)    \ar[rr]^{f'} \ar[dd]_{g'}  \ar@{}[ddrr]|{(PB2)}
						 && (S_1,l_1)  \ar[dd]^{g} \\~\\
						 && (S_2,l_2)  \ar[rr]_{f} 
						&& (S_0,l_0)
						}
						$
				 \caption{
	         \label{f.lSetPB} Pullback Construction in $\lSets$}
			 \end{figure}
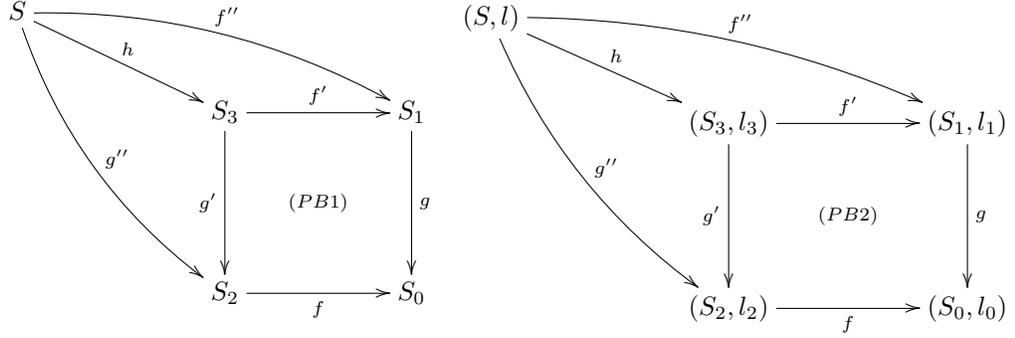
						
In $\lSets$ we define $l_3: S_3 \to (A,\le)$ with $l_3(s) = max\{ l_1\circ f' (s),l_2\circ g' (s)\}$. Obviously, $g'$ and $f'$ are then order-preserving:\\
$l_3(s)  = max\{ l_1\circ f'(s), \, l_2\circ g' (s)\} \ge l_1\circ f'(s) $;  the same for $g'$.\\
Moreover, $h:S\to S_3$ is also order-preserving: \\
We have $l(s) \ge l_1\circ f''(s)$ and $l(s) \ge l_2\circ g''(s)$,  so 
\begin{align*}
                        l(s) &  \ge   max\{ l_1\circ f''(s),\, l_2\circ g''(s) \}\\
                              & = max\{ l_1\circ f' \circ h(s),  \, l_2\circ g'  \circ h(s)\} \\
															& =l_3\circ h (s)
\end{align*}

Moreover \M \ is stable under pullbacks.\\
 $f\in \M$ implies $ l_2 = l_0 \circ f$, so we have for an arbitrary $s_3 \in S_3$  \\
$l_1\circ f'(s_3) \ge l_0\circ g\ circ f'(s_3)=  l_0\circ f \circ g'(s_3) = l_2\circ g'(s_3)$.\\
Hence, $l_3(s_3) =  max\{ l_1\circ f' (s),l_2\circ g' (s)\} =  l_1\circ f' (s)$ and $f' \in \M$.
\end{proof}

\begin{theorem}{$\lSets$  is an $\M$-Adhesive Category}
  \label{th.madTS.lSets}
\end{theorem}

\begin{proof}
	\begin{enumerate}
		\item The class $\M$ in $\lSets$  is  PO-PB compatible, since
		\begin{itemize}
			\item pushouts along \M-morphisms exist and \M \ is stable under pushouts, see Lemma~\ref{l.POlSets}
		  \item pullbacks along \M-morphisms exist and \M \ is stable under pullbacks , see Lemma~\ref{l.PBlSets} 
		  \item and obviously, \M \ contains all identities and is closed under composition.
	\end{itemize}
	\item  In $\lSets$ pushouts along \M-morphisms are $\M$-VK squares:\\ 
	 In $\lSets$ let be given a  pushout as (\ref{eq.madhPO})  in Def.~\ref{d.VK} with $m \in \M$  and some commutative cube as (\ref{eq.madhVK})  in Def.~\ref{d.VK} with (\ref{eq.madhPO})  being 
the bottom square and the back faces being pullbacks, then we have:
 
 \begin{description}
	 \item[$\Rightarrow$:] Let the top of (\ref{eq.madhVK})  in Def.~\ref{d.VK} be a pushout in $\lSets$. Pullbacks preserve  \M-morphisms, so  $m'\in \M$ and  
	      hence the top square is a pushout in $\cSets$ as well. As the category $\cSets$ is $\M$-adhesive, the front faces are pullbacks in $\cSets$ as well.
				Since the construction of pullbacks coincides in $\cSets$ and $\lSets$, the front faces are pullbacks in $\lSets$.
	
   \item[$\Leftarrow$:] Let the front faces be pullbacks in $\lSets$, and hence pullbacks in $\cSets$. Since $m \in \M$ (\ref{eq.madhPO})  in Def.~\ref{d.VK}
	       is pushout in $\cSets$ as well. So, $\cSets$ being adhesive, we have the top square being a pushout 
				in $\cSets$. Moreover, $m'\in \M$ as the back face is a pullback  preserving \M-morphisms. So, the top
				is a pushout along $\M$ is $\lSets$.
 \end{description}

\end{enumerate}
    Hence,  $(\lSets,\M)$ is an $\M$-adhesive category.
\end{proof}
Next we use Thm.~\ref{th.madTS.lSets} to prove that  place/transition nets with  label subtyping yield an $\M$-adhesive category.

\begin{definition}{Category of  place/transition nets with  subtyping of labels $\cPTs$}
   The category of  place/transition nets with  subtyping of labels $\cPTs$  is given by PT nets   $N = (P, T, pre, post, \pl, \tl, \mrk)$ over the alphabet $A=((A_P,\le_P),(A_T,\le_T))$ where 
	  $(P,\pl)$ is a labelled set over $ (A_P,\le_P)$ and $(T,\tl)$ is a  labelled set over $ (A_T,\le_T)$.
		net morphisms $f=(f_P,f_T): N_1 \to N_2$ where $f_P$ and $f_T$ are  order-preserving mappings.
\end{definition}

\begin{theorem}{$(\cPTs,\M)$ is an $\M$-adhesive category}
  \label{th.madTS.PTs}
  \end{theorem}
	\begin{proof}
The proof applies the construction for weak adhesive HLR categories (see Thm.~4.15 in \cite{FAGT}):
	 We know that  $(\lSets,\M)$  with $\M$ being the strict order preserving, injective mappings is an $\M$-adhesive
category and that $ (\_)^\oplus: \cSets \to \cSets$ preserves pullbacks along injective morphisms.
 As shown above $(\lSets,\M)$  with $\M$ being the strict order-preserving mappings is an $\M$-adhesive  
 category and  $G: \lSets \to \cSets$ preserves pushouts along $\M$-morphisms.
So, the category $\cPTs$ is isomorphic to the comma category $ComCat(G,(\_)^\oplus;I)$ with I = {1,2}, where $G: \lSets \to \cSets$ is the right adjoint  (see Lemma~\ref{l.adjunction}) from partial ordered sets to sets and $(\_)^\oplus$ is the free commutative monoid functor and hence  an $\M$-adhesive category.
	\end{proof}

	\begin{remark}
	\label{r.morePNcats}
	Further categories of Petri nets with subtyping of labels can be obtained using the constructions of previous papers  by replacing the category $\cSets$ by $\lSets$:
	\begin{enumerate}
		\item \label{morePNcats.i}  In \cite{Pad12} decorated place/transition nets yield an $\M$-adhesive  transformation category  $\cdecoPT$ for $\M$ being the corresponding class of strict morphisms, replacing  $\cSets$ by $\lSets$ we obtain   $\cdecoPTs$ for $\M$ being based on strict order preserving, injective mappings.		
		 \item Algebraic high-level nets  have been shown in  \cite{Prange08} to be an $\M$-adhesive  category $\cAHL$ for $\M$ being the class of strict morphisms. Replacing  $\cSets$ by $\lSets$ we obtain   $\cAHLs$ for $\M$ being based on strict order preserving, injective mappings. 
		\item Decorated place/transition nets with inhibitor arcs and  algebraic high-level nets with inhibitor arcs  also yield $\M$-adhesive categories (see \cite{Padberg14}) we can extend them with subtyping of labels as well.
		\item We can combine subtyping of labels even with transition priorities (see \cite{Padberg15}).
		A category of labelled partial orders, where the partial order is independent of the order of the labeling, is the basis and can be proven to be $\M$-adhesive.		
	\end{enumerate}
	\end{remark}

\subsection{Construction of the Name Space for Hierarchical Reconfigurable Petri Nets}\label{ss.namespace}
In \cite{Padberg15} the category of partial ordered sets $\poSets$ where the objects are partially orders sets and the morphisms are
order-preserving maps, that are maps $f : A \to B$ preserving the order, so $a\le a'$ implies $f(a) \le f(a')$. Here, we use for the name space
partial ordered sets with a greatest element $\g$ and the additional condition that $f(\g)=\g$.

	\begin{definition}{Category of partial ordered sets with a greatest element $\poSetsg$}
	The objects $(A,\le_A,\g)$ are partially orders sets with a greatest element $\g$ and the morphisms are
order-preserving maps $f : A \to B$ so that  $a\le a'$ implies $f(a) \le f(a')$ and $f(\g)=\g$.
\end{definition}
This category has obviously initial and final object $(\{\g\},\le,\g)$ and coproducts.
	The construction of pushouts is the same in  $\poSets$ (in \cite{Padberg15}), and pushouts of strict order
embeddings are pushouts in $\cSets$ as well.

For the construction of the name space for local and global rules we need an additional construction. it is an interesting question whether this corresponds to some standard (categorical) construction.
\begin{definition}{Name space $ (\mathfrak{A},\le,\g)$}
Given subsets $(A_i,\le_i,\g)$ for $i \in I$ of the the global name space $(A,\le,\g)$, then we have the coproduct $(C,\le_C,\g)  = \coprod_{i\in I} (A_i,\le_i,\g) $.
$$\xymatrix{
      &&  (A_1,\le_1,\g)   \ar@{^{(}->}[dddrr]|{inc_1}  \ar[dddll]|{c_1}\\
      &&  (A_2,\le_2,\g)   \ar@{^{(}->}[ddrr]|{inc_2}  \ar[ddll]|{c_2}\\
      &&  \vdots  \\
           (C,\le_C,\g) = \coprod_{i\in I} (A_i,\le_i,\g)   \ar[drr]|{c_C}
			&&  (A_i,\le_i,\g) \ar@{^{(}->}[rr]|{inc_i}  \ar[ll]|{c_i}
			&& (A,\le_A,\g)  \ar[dll]|{c_A}\\
			&& (\mathfrak{A},\le,\g)
}
$$
$(\mathfrak{A},\le,\g)$ is given as $\mathfrak{A} = C \coprod A$  in $\poSetsg$ and 
\begin{align*}
 \le \;\; & =   \{ (c_C \circ c_i (x), c_C \circ c_i (y)) \mid x \le_i y \text { for } i \in I\}   &  \text{ elements  of } \le_C \\
             &  \cup  \{ (c_A (x), c_A (y)) \mid x \le_A y \}  & \text{ elements of  } \le_A  \\
		         & \cup \{(c_C \circ c_i (x_i),c_A(x)) \mid   inc_i(x_i) = x  \text{ with } x_i \in A_i \text { for } i \in I\} &\text{  gobal  names greater than local ones}
\end{align*}

\end{definition}

\begin{example}{Name space construction}
In this example we have the global name space $A=\{a,b,c,d,e,f,g,\mathfrak{z}\}$ with the partial order give as a Hasse diagram and the greatest element $\mathfrak{z}$.
The subsets $A_i \subset A$ for $i=1,2,3$  denote the local name spaces. The coproduct  $(C,\le_C,\mathfrak{z}) = \coprod_{i\in I} (A_i,\le_i,\g)$ duplicates all elements except the greatest $\mathfrak{z}$, indicated by the indices. 

$\mathfrak{A}$ is then again the coproduct of $C$ and $A$, keeping the global names distinct from the local ones.
 Moreover,
$\le$ is the corresponding  union of the relations $\le_i$ and$\le_A$ with the additional relations that eaxg global name $x\in A$ is greater that the corresponding local ones
$x\ge x_i$:
\begin{figure}[H]

	\includegraphics[width=1.00\textwidth]{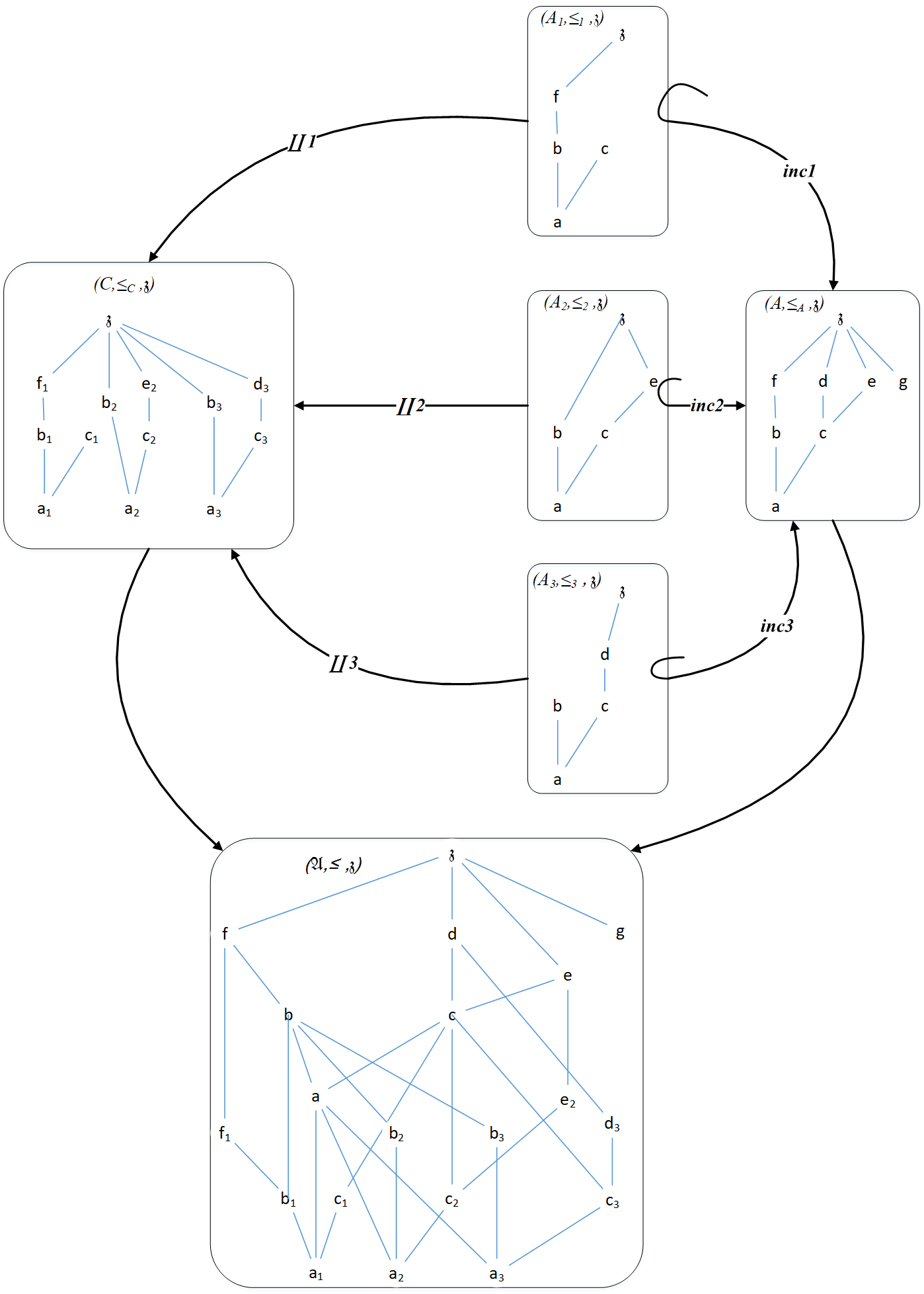}
	\caption{Example for name space construction}
	\label{fig:namespaceGlobal}
\end{figure}

\end{example}

	\begin{corollary}{Results} 
	These results  hold:
	\begin{itemize}
		\item Local Church Rosser Theorem for pairwise analysis of sequential and parallel independence
               \\ \hfill see Thm. 5.12 in \cite{FAGT}
\item Parallelism Theorem for applying independent rules and transformations in parallel 
               \\ \hfill see Thm. 5.18 in \cite{FAGT}
\item Concurrency Theorem for applying E-related dependent rules simultaneously 
               \\ \hfill see Thm. 5.23 in \cite{FAGT}
\item Embedding and Extension Theorem for transferring transformations and analysis results to more complex scenarios
               \\ \hfill see Thms. 6.14 and 6.16 in \cite{FAGT}
\item Local Confluence Theorem and Completeness of critical pairs for analyzing conflicts and for showing local Confluence 
               \\ \hfill see Thm. 6.28 and Lemma 6.22 in \cite{FAGT}

\end{itemize}

\end{corollary}

\section{Basics of Reconfigurable Petri Nets}
\label{s.basic}
In this section we give the basic notions. Note that in \reconnet the underlying type of nets are decorated
place/transition nets.

  We use the algebraic approach to Petri nets, where the pre- and post-domain functions $pre,post: T \to P^\oplus$
 map the transitions $T$ to a multiset of places $P^\oplus$ given by the set of all linear sums over the set $P$.  A marking  is given by $m \in P^\oplus$ with $m= \sum_{p\in P} k_p\cdot p$. 
The $\le$ operator can be extended to linear sums: For  $m_1, m_2 \in P^\oplus$ with $m_1 = \sum_{p\in P} k_p\cdot p$ and $m_2 = \sum_{p\in P} l_p\cdot p$ we have $m_1 \leq m_2$  if and only if $k_p \leq l_p$ for all $p \in P$. The operations {\grqq +  \grqq} and {\grqq -- \grqq} can be extended accordingly.
\komtar{nur was dann auch gebraucht wird}{}

Here, we introduce reconfigurable place/transition nets with labels and subtyping of labels for the rules. These labels need a name space that  is given by a partial order $(A,\le, \g_A)$ with a greatest element,  $a\le \g_A$ for all $a \in A$.

\begin{definition}{Labelled place/transition  nets} \label{d.PT}
	A (marked  labelled place/transition) net is given by $N = (P, T, pre, post, \pl, \tl, \mrk)$  over the namespace $A=(A_P,A_T)$ with partial orders $(A_{P},\le_A,\g_p)$ and $(A_{T},\le_T,\g_T)$.
	$P$ is a set of places,  $T$ is a set of transitions. $pre : T \to P^\oplus$ maps a transition to its $pre$-domain  and $post : T \to P^\oplus$ maps it to its $post$-domain. Moreover, $\pl: P \to (A_{P},\le_A,\g_p)$ is a label function mapping places  to a name space,
  $\tl : T \to (A_{T},\le_T,\g_T)$ is a label function mapping  transitions to a name space
		and $\mrk \in P^\oplus$ is the marking denoted by a multiset of places.
		
A transition $t \in T$ is $\mrk$-enabled for a marking
$\mrk \in P^\oplus$ if we have  $pre(t) \leq \mrk$.  The follower
marking $m^\prime$ is computed by $\mrk^\prime=\mrk - pre(t) + post(t)$ and represents   the result of a  firing step $\mrk\fire{t} \mrk^\prime$. 
	\end{definition}
	
	The labelling function is provided with an order for subtyping, this allows more abstract rules that can be applied for occurrences with lesser labels, for an example see Sect.\ref{s.intEx}.

	A reconfigurable Petri net $RN=(N,\R)$ consists of a Petri net $N$ and a set of  rules $\R$.  This  allows  reconfigurable Petri nets to modify themselves. Rules are defined by a span  of net morphisms 
   $ r  = (L \gets K \to R)$
where $L$ is the left-hand side and  $K$ is an interface between $L$ and $R$ the right-hand side. 
The basic idea is to find $L$ in the net $N$ and  replace it by $R$.  An occurrence morphism $o:L \to N$ is required to identify the relevant parts of the left-hand side $L$ in $N$.

Net morphisms  are given as a pair of mappings for the places and the transitions preserving the structure, the labels and the marking. Given two   nets $N_1$ and $N_2$ as in Def.~\ref{d.PT}  a net morphism $f:N_1 \to N_2$  is given by $f=(f_P:P_1 \to P_2,f_T:T_1 \to T_2)$, so that $pre_2 \circ f_T = f_P^\oplus \circ pre_1$ and $post_2 \circ f_T = f_P^\oplus \circ post_1$ and $m_1(p) \le m_2(f_P(p))$ for all $p \in P_1$. 
	The labels are mapped so that  $\tl_2 \circ f_T (t)  \le \tl_1(t)$ for all $t\in T_1$ and $\pl_2 \circ f_p (p)  \le \pl_1(p)$ for all $p\in P_1$.

$$ 
	\xymatrix{
	         && T_1 \ar@<1mm>[rr]^{pre_1}  \ar@<-1mm>[rr]_{post_1}  \ar[dll]^{\tl_1} 
					          \ar[dd]|{f_T}
	         &&  {P_1}^\oplus  \ar[drr]^{\pl_1}   
					          \ar[dd]|{{f_P}^\oplus}\\
	               (A_T,\le_T,g_T)
	&&&&&& (A_P,\le_P,g_P)\\
	         && T_2 \ar@<1mm>[rr]^{pre_2}  \ar@<-1mm>[rr]_{post_2}  \ar[ull]^{\tl_2} 
	         &&  {P_2}^\oplus \ar[urr]^{\pl_2}\\
	}
	$$

Moreover, the morphism $f$ is called strict if  both  $f_P$ and $f_T$  are injective, if $\tl_2 \circ f_T  = \tl_1$  and $\pl_2 \circ f_p= \pl_1$, and if
 $ m_1(p) =m_2(f_P(p))$  holds for all $p \in P_1$. 

 A transformation step 
$N \deriv{(r,o)} M$ via rule $r$  can be constructed in two steps by the  commutative squares (1) and (2) in Fig.~\ref{dpo}.
\label{glue} Given a rule with an occurrence $o:L\to N$ the \textit{gluing condition}  has to be satisfied in order to apply a rule at a given occurrence. Its satisfaction requires that the deletion of a place implies the deletion of the adjacent transitions, and that the deleted place's marking does not contain more tokens than the corresponding place in $L$. 

\begin{figure}
	\centering
$
\xymatrix{  L  \ar[d]|o \ar@{}[dr]|{\bf (PO1)}   
            & K \ar[l] \ar[r] \ar[d]   \ar@{}[dr]|{\bf (PO2)}
            & R \ar[d]\\
             N
            & D \ar[l] \ar[r]
            & M
              }
 $
\caption{\label{dpo}  Net transformation}
\end{figure}
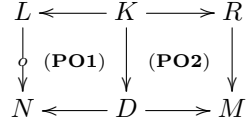

A reconfigurable Petri net $N$ can either fire an activated transition or execute a transformation step$N \xRightarrow{\text{(r,o)}} M$.
Figure \ref{dpo} illustrates the transformation of a net using two pushouts $(PO1)$ and $PO(2)$.

This is possible because  nets with labels and subtyping  can be proven to be an $\M$-adhesive  category, see Sect.~\ref{s.subtype}.
Hence these results  hold for the corresponding type of labelled Petri net:
	\begin{itemize}
		\item Local Church Rosser Theorem for pairwise analysis of sequential and parallel independence
               \\ \hfill see Thm. 5.12 in \cite{FAGT}
\item Parallelism Theorem for applying independent rules and transformations in parallel 
               \\ \hfill see Thm. 5.18 in \cite{FAGT}
\item Concurrency Theorem for applying E-related dependent rules simultaneously 
               \\ \hfill see Thm. 5.23 in \cite{FAGT}
\item Embedding and Extension Theorem for transferring transformations and analysis results to more complex scenarios
               \\ \hfill see Thms. 6.14 and 6.16 in \cite{FAGT}
\item Local Confluence Theorem and Completeness of critical pairs for analyzing conflicts and for showing local Confluence 
               \\ \hfill see Thm. 6.28 and Lemma 6.22 in \cite{FAGT}

\end{itemize}

\section{Hierarchies of Nets and Rules}
\label{s.hier}
 A \HRPN \   uses  \textit{substitution transitions} to implement the hierarchy. A substitution transition is a special kind of transition that itself does not fire,  instead it contains a subnet that defines the behavior that takes place in its stead.
Following this basic definition of substitution transitions, different implementations suited for specific purposes are possible, this work focuses on the variant of the substitution transition based hierarchical Petri net that have been presented in \cite{jensen2009}.

Each substitution transition has its own subnet with its own local rules.
All places that share an edge with a substitution transition are called the transition's  \textit{connecting places}. For each connecting place of the substitution transition there exists a corresponding connecting place in the transition's subnet  with the same marking.
Via these places tokens enter and leave the subnet. A transition that  fires  is from either the main net or some subnet, but no substitution transition. Any net can contain multiple substitution transitions  each instantiating exactly its own subnet. Although multiple substitution transitions may instantiate the same subnet layout, each substitution transition has it's own permanent instance.
This leads to a behaviour of the main nets that relies solely on the firing of the subnets, i.e the firing of the flattened net.

\begin{figure}[H]\label{fig:PFPP}
\centerline{\includegraphics[width=0.40\textwidth]{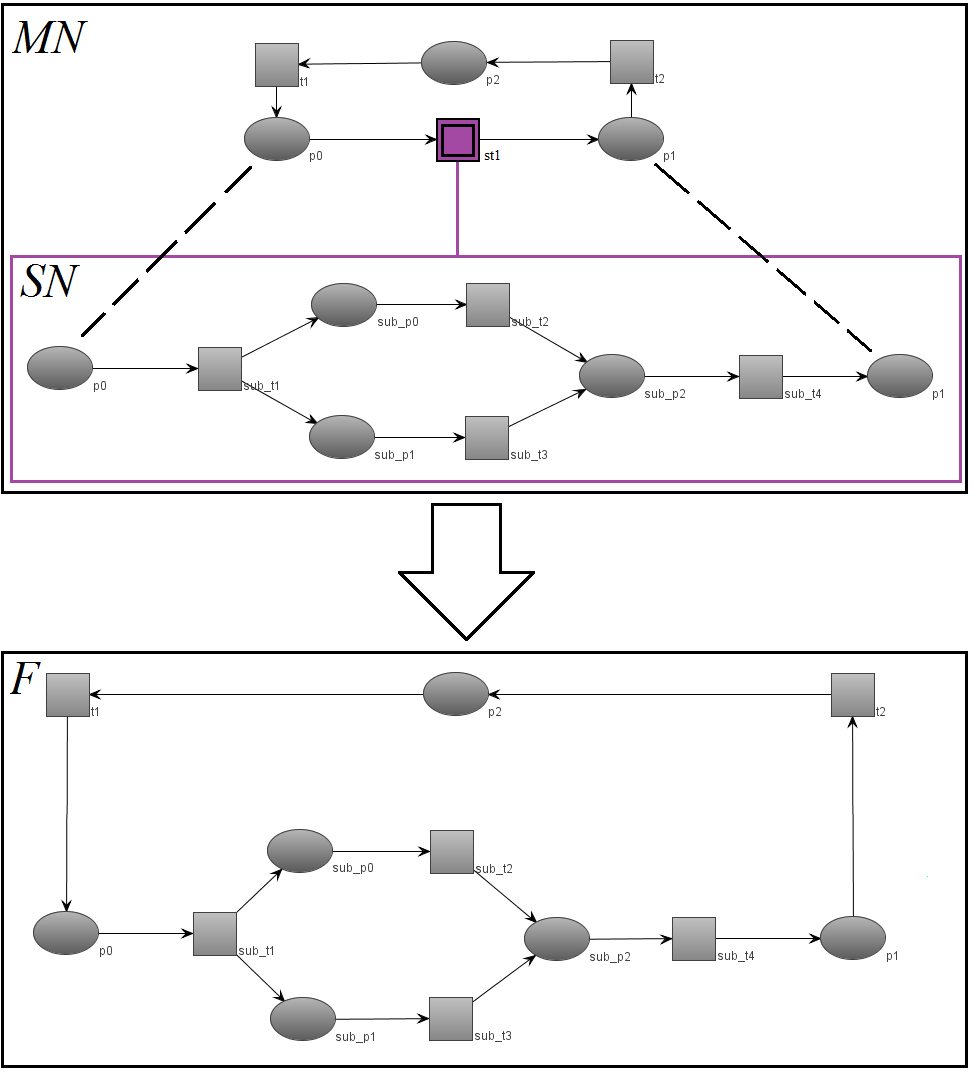}}
\caption{\label{fig_PFPP}{Flattening of a substitution transition. }}
\end{figure}

Figure \ref{fig_PFPP} shows in the top half a hierarchical net with it's main net $MN$ and a subnet $SN$. In the main net the substitution transition $st1$ has two connecting places: $p0$ has an edges connecting it to $st1$ and $st1$ has an edge connecting it to the place $p1$. These places can also be found in the subnet as connecting places with edges to and from different transitions. If tokens are added to the place $p0$ via the transition $t1$ these  also appear in the subnet. There $SN$s transition $sub\_t1$ can fire and remove tokens from $p0$ resulting in the removal of the same tokens from $p0$ in $MN$. 

Subnets may contain substitution transitions containing further subnets resulting in a nested hierarchy.


 We have \textit{local rules} and  \textit{global rules}. Local rules are given for a subnet only, whereas global rules belong to the hierarchical net and can be applied in all subnets since their labels are greater than the labels in the subnet. For details see Subsection~ \ref{ss.namespace}. The name space is given by the disjoint union of all local name spaces, so that local rules can be applied only with in the given subnet. 
Local rules  respect the  hierarchical net borders that means no transformation  may effect more than one (sub-)net. Hence, one  restrictions is imposed on the  rules: Substitution transitions may not  be part of a  rule. As a consequence connecting places may not be deleted or added by a  rule, but they can be part of one. since connection places are neighbours of substitution transitions that cannot occur in a rule, they can be neither added nor deleted.

The definition of the reconfigurable hierarchical Petri net requires the substitution transition together with its adjacent places, called  net $Net(t)$ of a transition $t$.
\begin{definition}{$Net(t)$} \label{nt}
Given $N=(P,T,pre,post,\pl,t_{name})$ then for a transtion $t\in T$  the net of $t$ is the net 
$Net(t)=(^{\bullet}t \cup t^{\bullet}, t, pre_{|t}, post_{|t}, p_{name_{|^{\bullet}t \cup t^{\bullet}}}, t_{name_{|t}})$.
\end{definition}
With this reconfigurable hierarchical Petri nets can be formally defined.
\begin{definition}{Hierarchical reconfigurable Petri net} \label{d.hPN}
  A \HRPN \ $HN=(RN, A, GR)$ is given by a reconfigurable net with substitutions $RN=(N,\R^N, {SR}^N)$, a name space $A=(\A{P}{},\A{T}{})$ and a set of global rules $GR$ over $A=(\A{P}{},\A{T}{})$, so that 
	
	\begin{itemize}
		\item  $N = (P, T, pre, post, \pl,\tl, M)$ is a place/transition net  over  $(\A{P}{N},\A{T}{N})$ so that
	\begin{itemize}
		  \item  $P$ is a set of places. 
      \item $T$ is a set of transitions that  contains substitution transitions $sT \subseteq T$.
      \item $pre:T\to P^\oplus $ is a function used for the pre-domain of each transition.  
      \item $post:T\to P^\oplus $ is a function used for the  post-domains of each transition.      
      \item $\tl: T \xrightarrow{} \A{T}{N}$  is a naming function for transitions,  where
			         substitution transitions have their own name space $A_{sT}\subseteq \A{T}{N}$ 
							 so that   $\tl(sT) \subseteq A_{sT} $ and  injective $\tl_{|sT}$. Moreover 
	                 $\tl(T \backslash sT) \subseteq A_T^{RN} \backslash A_{sT}$.
      \item $\pl: P \xrightarrow{} \A{P}{N}$ is a naming function for places, where
									the set of connecting places $cP =\{ \bullet t \cup t \bullet \mid \, t \in sT\} \subseteq \P$is  given 
									by the neighbourhood of the substitution transitions and  the name space of the connecting places 
									 $A_{cP} \subseteq \A{P}{N}$ satisfies   $ \pl(cP) \subseteq \A{P}{N}$.									
      \item $M$  is a set of tokens by $M \in P^\oplus$.
	\end{itemize}
	\item $\R^N$is a set of local rules over $(\A{P}{N}, \A{T}{N}\setminus A_{sT})$.
	\item $SR^N$ is a set of substitution rules together with a mapping if substitution transition to substitution rules $subst: sT \to SR^N$
	 so that $subst(t) = (Net(t)   \xleftarrow{}  CP(t) \xrightarrow{} SN^t)$  with
\begin{itemize}
	\item the interface $CP(t) = (^{\bullet}t \cup t^{\bullet}, \emptyset, \emptyset, \emptyset, \pl_{|{^\bullet}t \cup t^{\bullet}},\emptyset$ consisting of connecting  places only.

   \item a reconfigurable net with substitutions $SN^t =( RN^t, \R^t, SR^t)$ over $A^t=(\A{P}{t},\A{T}{t})$ with
	  $A_{cP} \subseteq \A{P}{t}$.
\end{itemize}
	\end{itemize}
\end{definition}

Figure \ref{fig_FNR} shows an example for a very basic substitution rule.

\begin{figure*}[ht]\label{fig:FNR}
\centerline{\includegraphics[width=\linewidth]{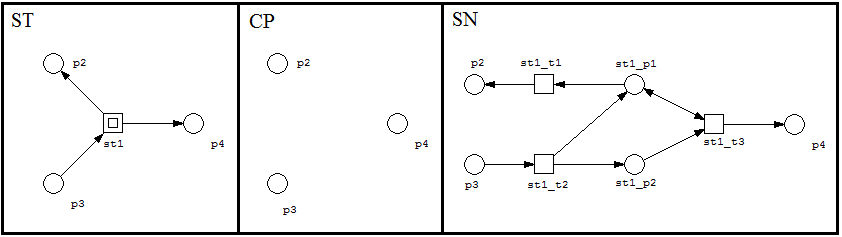}}
\caption{\label{fig_FNR}{ An exemplary basic substitution rule.  }}
\end{figure*}

\cite{jensen2009} Chapter 5 states that  the flattening of a hierarchical net that uses substitution transitions must remove each substitution transition    and insert its subnet  into the supernet by fusing the connecting places. 
\komtar{BSP aus PJ2?}{}
This process corresponds to applying the substitution rules from Definition \ref{hpndef}. 
Only one substitution for each substitution transition  is applicable to $RN$.

\begin{corollary}{Set of substitions $S^N$}
Given a reconfigurable net with substitutions $RN=(N,\R^N, {SR}^N)$.
For every substitution transition $t \in sT$ and its substitution rule  $sr= subst(t)$ there exists exactly one injective occurrences $o$ of $sr$ 
These substitutions are collected in a set of subsitutions $S^N =\{(sr,o)\mid  sr= subst(t) \text{ and } o: Net(t) \hookrightarrow RN\}$.
\end{corollary}

 Due to the global and local rules  flattening construction is more complex than for a normal hierarchical Petri net. Flattening of a normal hierarchical Petri net looses all information of the hierarchical borders. but this information is needed for the correct application of local and global  rules in the flattened net.

First we investigate the parallel independence \cite{FAGT} of the subsitution rules.
\begin{lemma}{Pairwise Independence of Substitutions}\label{sindislemma}

Given a reconfigurable net with substitutions $RN=(N,\R^N, {SR}^N)$. Any two substitutions $s_1, s_2 \in S^N$ are pair-wise independent from one another if $s_1 \neq s_2$.
\end{lemma}
If any two $s_1, s_2 \in S$ with $s_1 \neq s_2$ are pairwise parallel independent, with the help of the Local Church-Rosser Theorem, it can be deducted that they are also sequentially independent \cite{FAGT}. All substitution rules $sr$ together with their  occurrences are independent from another if  any two $sr_1,sr_2$ with $sr_1 \neq sr_2$ are pairwise independent. So the proof of parallel independence of two arbitrary substitutions $s_1,s_2 \in S$ is sufficient to prove Lemma \ref{sindislemma}.  
\begin{proof}\label{sindis} 
We show for two arbitrary  $s_1 \neq s_2$   the set theoretic representation of parallel independence $ o_{1}(ST_1) \cap o_{2}(ST_2) \subseteq o_{1}(l_1(CP_1)) \cap o_{2}(l_2(CP_2)) $.
%
 %
$$\xymatrix{
        SN^{t_1}  \ar[d]|{n_1}
    &  CP(t_1)     \ar[d]|{k_1}   \ar[l]|{r_1} \ar[r]|{l_1}
		&  Net(t_1)   \ar[dr]|{o_1}  \ar@/^2mm/@{-->}[rrrd]|(0.6){g_1}
  &&  Net(t_2)    \ar[dl]|{o_2}  \ar@/_2mm/@{-->}[llld]|(0.6){g_2}
    &  CP(t_2)      \ar[d]|{k_2}   \ar[l]|{r_2} \ar[r]|{l_2}
     & SN^{t_2}   \ar[d]|{n_2}\\
		    H_1
		&  D_1  \ar[l]|{r'_1} \ar[rr]|{l'_1}
	 && G
	&&  D_2   \ar[ll]|{r'_2} \ar[r]|{l'_2}
	  &  H_2
}
$$
The left-hand side of any rule $rs$ of  $(rs,o) \in S^N$ contains by definition \ref{d.hPN} only a net $Net(t)$. As specified in Def.~\ref{nt} $Net(t)$ contains only a substitution transition $t$ and $t$'s pre- and post-domains. The interface $CP(t)$ contains only $t$'s pre- and post-domains.  
Considering two substitutions $s_1, s_2 \in S^N$ with $s_1 \neq s_2$, the intersection between their occurrences only considering transitions must be empty because otherwise $t_1 = t_2$ and thus $s1 = s2$. Since $CP(t_1)$ only contains places and since $Net(t_1)$ contains one distinct transition $t_1$ and $Net(t_2)$ the another one $t_2$, it follows:
$o_{1T}(Net(t_1)) \cap o_{2T}(Net(t_2)) = \emptyset  \subseteq \emptyset = o_{1T}(l_1(CP(t_1))) \cap o_{2T}(l_2(CP(t_2)))$

Now we  consider the  places. Let
$ p \in o_{1P}(Net(t_1)) \cap o_{2P}(Net(t_2))$. Hence $p \in (^{\bullet}t_1 \cup t_1^{\bullet}) \cap (^{\bullet}t_2 \cup t_2^{\bullet})$
that is $p \in CP(t_1) \cap CP(t_2)$ by  definition of $CP$.\\
Since  $l_1, l_2, o_{1P}$ and $o_{2P}$ are functions we have $p \in (l_1(CP(t_1))) \cap (l_2(CP(t_2)))$
 and $p \in o_{1P}(l_1(CP(t_1))) \cap o_{2P}(l_2(CP(t_2)))$.
Thus: $
o_{1P}(Net(t_1)) \cap o_{2P}(Net(t_2)) \subseteq o_{1P}(l_1(CP(t_1))) \cap o_{2P}(l_2(CP(t_2))) $
which proves any two $s_1, s_2 \in S$ with $s_1 \neq s_2$ are pairwise parallel independent. 
\end{proof}

With the help of Lemma \ref{sindislemma} now  Theorem \ref{fwelldeftheorem} can be proven.
\begin{theorem}{$FLAT(N,SR^N)$ Flattening of reconfigurable net with substitutions}\label{fwelldeftheorem}
Given a reconfigurable net with substitutions$(N,SR^N)$ any possible  transformation sequence of rules $S^N$  yields the same (up to isomorphism) well-defined net $N \deriv{S} FLAT(N,SR^N)$.
\end{theorem}
\begin{proof}
With all $s \in S$ being mutually independent, \cite{rozenberg1997} states all the transformation sequences $HN \xRightarrow{*} F$  are equivalent and there exists a parallel transformation sequence $HN \xRightarrow{\sum\nolimits_{s \in S} s} F$. Then we define  $FLAT(N,SR^N):=F $.
Such a parallel transformation sequence can always be constructed and is unique up to isomorphism. 
\end{proof}

The flattening of a \HRPN \ to a reconfigurable nets needs to include global as well as local rules and is given recursively based on  flattening of  nets with substitution.

\begin{definition}{Flattening}
The flattening is defined for an  hierarchical net $HN=( RN, A, GR)$  given by a reconfigurable net $RN=(N,\R^N, {SR}^N)$, an name space $A=(\A{P}{},\A{T}{})$ and a set of global rules $GR$ as given in  Def.~\ref{d.hPN} recursively by:
\begin{enumerate}
	\item   Given  $RN=(N,\R^N, {SR}^N)$ over $A=(\A{P}{N},\A{T}{N})$ with  substitution transitions $sT = \emptyset$ we have: \\
	            $flat(RN) =  (N,\R^N)$   over $A$
							
	\item    Given  $RN=(N,\R^N, {SR}^N)$ over $A=(\A{P}{N},\A{T}{N})$ with  substitution transitions $sT \neq \emptyset$ we have: \\
	           $flat(RN) = flat(FLAT(N,SR^N), \overline{\R},\overline{SR})$  over $\overline{A}$ with
	\begin{itemize} 
		\item $\overline{A} =  \biguplus_{t\in sT} ( \A{P}{t}\setminus A_{cp})  \uplus   A_{cP}$
		\item $\overline{\R}= (\biguplus_{t\in sT} \R_t) \uplus \R^N$
		\item $\overline{SR}= (\biguplus_{t\in sT} SR^t$
	\end{itemize}
	\item $flat(HN) = (N_{Flat},GR\cup \R_{Flat})$ over $\mathfrak{A}$ with $flat(RN) = (N_{Flat},\R_{Flat})$ where the name space $\mathfrak{A}$ ist the
	  union of the name spaces, so that the global  labels are greater than the corresponding local  labels (see Subsect.~\ref{ss.namespace}).
\end{enumerate}
\end{definition}

\begin{definition}{Well-defined \HRPN} \label{d.wellhPN} 
A \HRPN \ $HN=(RN, A, GR)$ is well-defined if and only if the $flat(HN)$ is well-defined.
\end{definition}

\subsection{Introductory  Example}
\label{s.intEx}

Reconfigurable Petri nets extend normal Petri nets to include the ability for dynamic change. This is achieved through the use of a rewriting system in the form of rules for the transformation of the net. This allows the modification of the net's structure at run time, which can be used in the modelling of dynamic reconfigurable hardware like FPGAs or flexible manufacturing systems.
When modelling such a system two kinds of changes need to be included, for one a change of state accomplished through the firing of Petri net transitions, but also the process itself can experience changes for which the rule based rewriting system is used.

Imagine some simple  but adaptive process that can alternatively execute three different tasks \texttt{task1}, \texttt{task2}, and \texttt{task3}. An abstract view of this process is given in Fig.~\ref{fig.AN}. 
\begin{figure}[H]
\begin{minipage}[b]{.49\textwidth}
		\includegraphics[width=\linewidth]{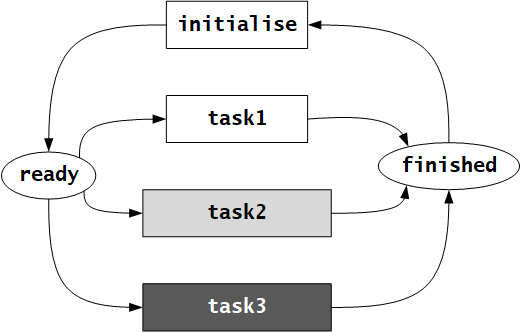}
	\caption{Abstract view of process: Net $AN$\label{fig.AN}}
	\end{minipage}
\hfill
\begin{minipage}[b]{.49\textwidth}
	\includegraphics[width=\linewidth]{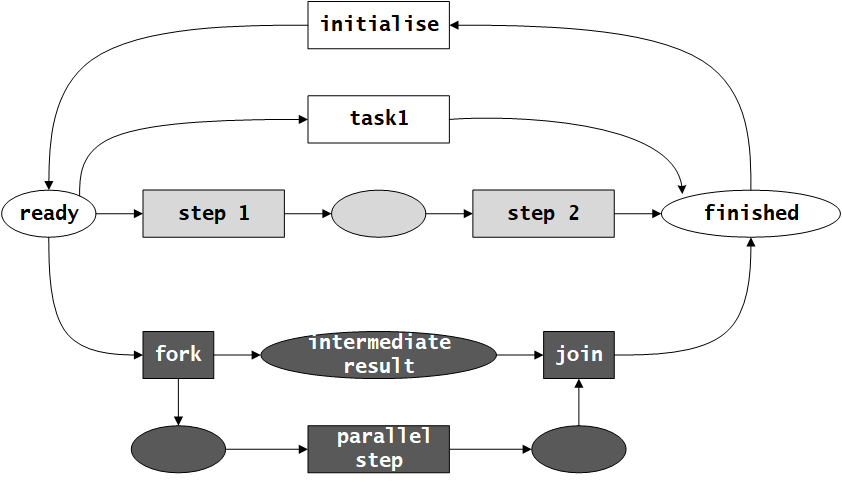}
	\caption{Flattened net\label{fig.FN}}
\end{minipage}
\end{figure}
The tasks  \texttt{task2} and \texttt{task3} are more complex and are 
given by subnets, where \texttt{task2} is a sequence of steps and \texttt{task3} includes some forking. The hierarchy concept in Sect.~\ref{s.hier} allows the substitution of the transitions with the subnets. The substitution of the transition \texttt{task2} replaces the transition and its adjacent places, that is $Net(\texttt{task2})$,  by the subnet $SN1$ and \texttt{task2} is replaced by $SN2$, both  in Fig.~\ref{fig.SN}. Applying these substitutions  to the abstract nets in  Fig.~\ref{fig.AN} yields the flatted net in Fig.~\ref{fig.FN}. 
\begin{figure}[h]
	\includegraphics[width=\linewidth]{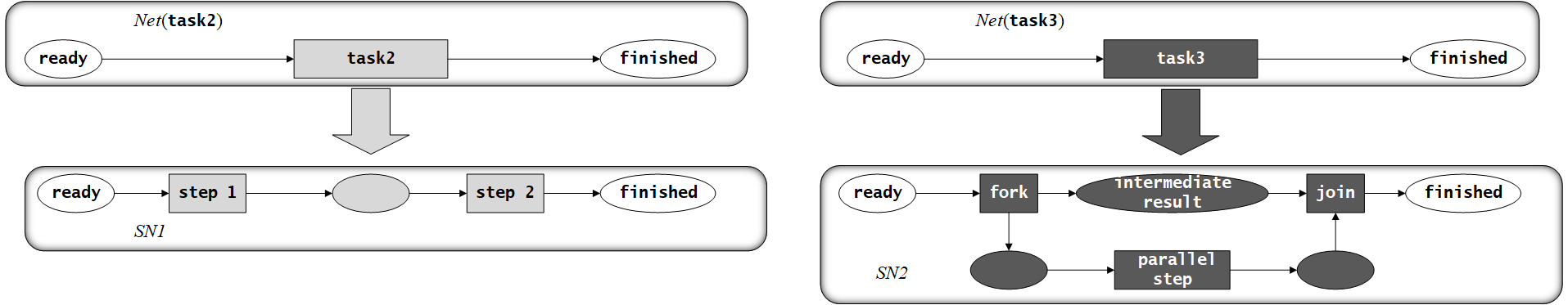}
	\caption{Substitution of transitions\label{fig.SN}}
\end{figure}

Now we add rules for the subnets for the adaptation of the tasks: \texttt{task1} is so simple, it requires no adaptation. In \texttt{task2}  the sequence of steps can be changed (rules \texttt{SN1:r1} and \texttt{SN1:r2}) or an intermediate steps is introduced  or removed ( rules \texttt{SN1:r3} and \texttt{SN1:r4}). So we have the four rules given in  light grey in Fig.~\ref{fig.locRules}.
In \texttt{task3} the intermediate step can be adapted  by rule \texttt{SN2:r5} so that parallel step may require something from the intermediate result. And this adaptation can be reversed by rule \texttt{SN2:r6}. both rules are given in dark grey in Fig.~\ref{fig.locRules}. These six rules are local rules, that should be only applied in the corresponding subnet.
\begin{figure}[H]
\centering
	\includegraphics[width=.9\textwidth]{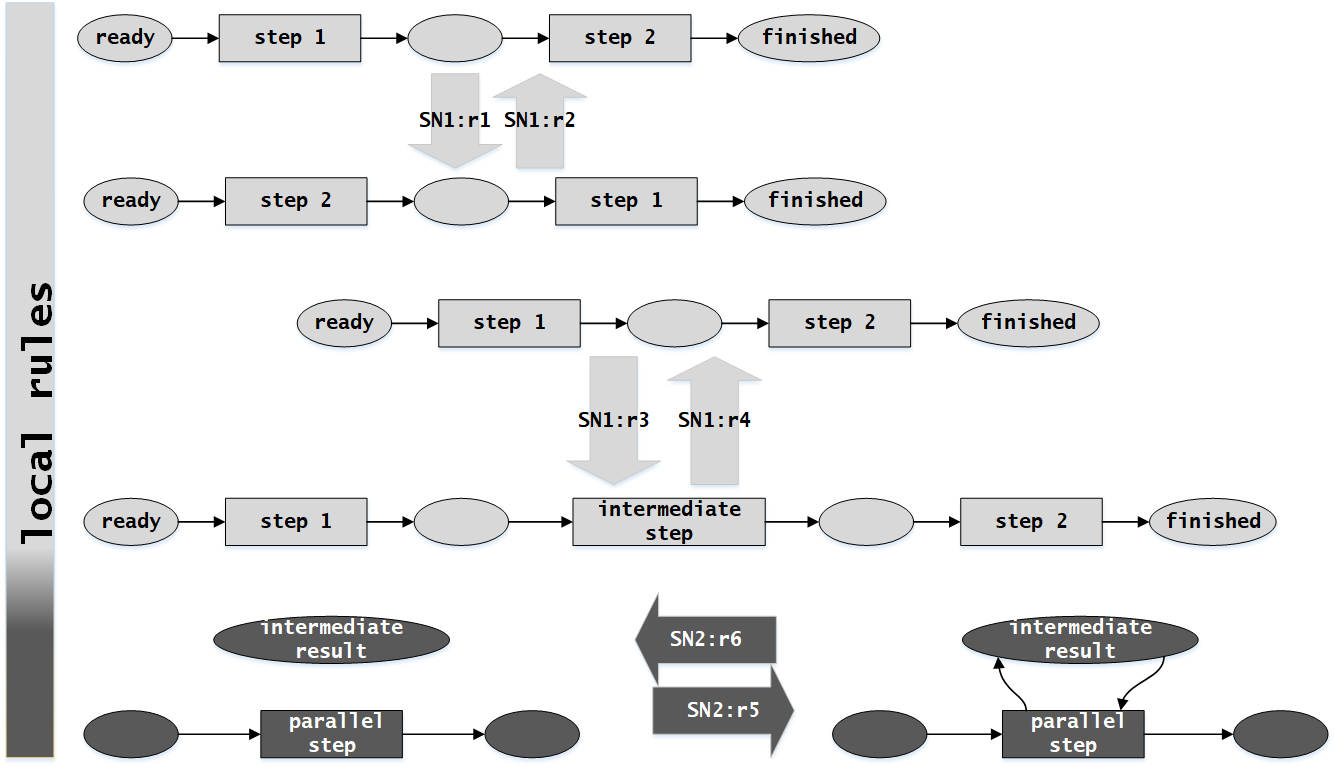}
	\caption{Local rules for the subnets $SN1$ and $SN2$}
	\label{fig.locRules}
\end{figure}

We have for the transitions the name space $A_T =\{\texttt{initialise}, \texttt{task1},\texttt{task2},\texttt{task3},\texttt{fork}, \texttt{join},$\\$ \texttt{step}, \texttt{step1},\texttt{step2}, \texttt{intermediate step}, \texttt{parallel step} \}$ that ensures the locality of the rules by the labels. 

Additionally, we want a global rule that adds to all steps a counting place. This rule is given below in Fig.~\ref{fig.globRule}. This rule can be applied at each transition with a lesser label.
The name space for the transition is ordered in the following way:\\
$\g_T \ge l$ for all $l \in A_T$ and\\
 $\texttt{step} \ge l$ for all $l\in \{ \texttt{step1},\texttt{step2}, \texttt{intermediate step}, \texttt{parallel step}\}$
\begin{figure}[H]
\centering
	\includegraphics[width=.9\textwidth]{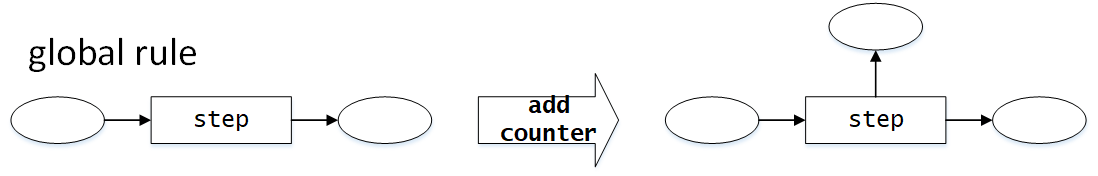}
	\caption{Global rule for adding counter}
	\label{fig.globRule}
\end{figure}

\section{Hierarchies in \reconnet}
\label{s.hier_reconnet}

During the simulation \reconnet's simulation engine  uses the flat representation of a \HRPN  \ for transition firing and transformation rule application, because this allows usings  \reconnet's simulation engine to handle the hierarchical net, i.e its flattened net.
However,for the user this will remain transparent and the visual interface will remain in a hierarchical view. While transitions are fired and transformations are made on the flat net the hierarchical view visualized the changes appropriately.
During the design phase of a \HRPN, in which the net designer develops the nets and transformation rules, true hierarchy is used and at the beginning of the simulation the flat net is acquired with the flattening process.
The application of local rules in the flat net needs one single name space for places and transitions ($A_P, A_T$). This name space needs to include all of the disjoint name spaces of the (sub-)nets. 
This single name space is created during the flattening. Whenever a subnet is inserted into its supernet all places and transitions that are not connecting places get a prefix to their names that is unique to the substitution transition that was replaced.
This way the naming preserves hierarchy borders and (sub-)net identities and so the names of places and transitions are specific enough that a rule meant for only a specific (sub-)net can be limited to the correct part of the flat net.
For persistence of a \HRPN \ from \reconnet  the \HRPN's flat net and the substitution rules.
The \HRPN is saved as a tuple of the main net, as a reconfigurable Petri net, its substitution rules and the flat net, so that $HN = \langle RN, SR, Flat(RN, SR) \rangle$.
So, the flat net can  be loaded directly and needs not to be computed each time again.

\subsection*{Flattening in \textit{ReConNet}}

In \textit{ReConNet} the flattening process can be realized as transformation unit \cite{kreowski2008graph} $HN \xRightarrow{\text{sr!}} F $ with $!$ as long as possible with injective occurrences. For the transformation unit an applicable substitution rule $sr$ with an occurrence is randomly picked and applied, this step is repeated until there no longer exists a $sr \in HN$ with an occurrence. 

\begin{lemma}[$HN \xRightarrow{*} F$ Produces a Well-defined Net F] \label{fwelldef}

The resulting F of  $HN \xRightarrow{\text{sr!}} F $ is well-defined up to isomorphism.
\end{lemma}

\begin{proof}

With any two substitution rules $sr_1,sr_2 \in SR$ being pairwise independent, $F$ being well-defined up to isomorphism can be proven with an indirect approach:

If $F$ is not well-defined the transformation sequences $HN \xRightarrow{\text{sr!}} F $ and $HN \xRightarrow{\text{sr!}} \widehat{F} $ exists so that $F \not \equiv \widehat{F} $.
For this to be true there has to exists some $M$ so that:
\begin{equation}\label{dia_Fwell_1}
\begin{gathered}
\xymatrix{
	&&&& M_j\ar@{=>}[drr]^{\ast}|{\SelectTips{cm}{12}\object@{/}}|{}\\
	HN \ar@{=>}[rr]^{*} && M \ar@{=>}[urr]^{s_j} \ar@{=>}[drr]_{s_i} &&&&F\\
	&&&& M_i\ar@{=>}[urr]_{\ast}|{\SelectTips{cm}{12}\object@{/}}|{ }
}
\end{gathered}
\end{equation}
Since for $HN \xRightarrow{\text{sr!}} F $ and $HN \xRightarrow{\text{sr!}} \widehat{F} $ both substitutions $s_i$ and $s_j$ have to be applied, all $s \in S$ are pairwise sequential independent and any sequence of sequentially independent transformations can be applied in arbitrary order, yielding the same well-defined resulting net \cite{FAGT}, Diagram \ref{dia_Fwell_1} can be written as:
\begin{equation}\label{dia_Fwell_2}
\begin{gathered}
\xymatrix{
	&&&& M_j\ar@{=>}[drr]^{s_i}|{\SelectTips{cm}{12}\object@{/}}|{}&\\
	HN \ar@{=>}[rr]^{*} && M \ar@{=>}[urr]^{s_j} \ar@{=>}[drr]_{s_i} &&&& M_{ij}\ar@{=>}[rr]^{*}&&F\\
	&&&& M_i\ar@{=>}[urr]_{s_j}|{\SelectTips{cm}{12}\object@{/}}|{ }&
}
\end{gathered}
\end{equation}
Any two $s \in S$ are pairwise parallel independent, so are $s_i$ and $s_j$, thus their sequence is interchangeable $s_i$ is always applicable to $M_j$ and $s_j$ is always applicable to $M_i$ both always leading to the same net $M_{ij}$. So $F \equiv \widehat{F} $ for $HN \xRightarrow{\text{sr!}} F $ and $HN \xRightarrow{\text{sr!}} \widehat{F} $ and thus $F$ is well-defined up to isomorphism.
\end{proof}

\begin{theorem}[Equivalence of Transformation Unit Application and Flattening Process]

The transformation via transformation unit is equivalent to the flattening process from Definition \ref{flatdef}. So that from $HN \xRightarrow{\sum\nolimits_{s \in S} s} F$  and $HN  \xRightarrow{\text{sr!}} \widehat{F}$ follows $F \equiv \widehat{F}$.
\end{theorem}

Since all $s \in S$ are independent from another, each $sr$ of the transformation unit  $HN  \xRightarrow{\text{sr!}} \widehat{F}$ can be applied at least once for each of its occurrences.
$HN \xRightarrow{\sum\nolimits_{s \in S} s} F$  can equivalently be applied as a transformation sequence $HN \xRightarrow{*} F = HN \xRightarrow{s_1} ... \xRightarrow{s_n} F$. So to prove that $F \cong \widehat{F} $ it is to show that no $sr$ under an occurrence $o$ can be applied more than once. Since $S$ contains all $sr$ with all their occurrences $o$, it is only to show that each $s \in S$ can be applied no more than once.

\begin{proof}

For any $s \in S$ to be able to be applied more than once it would have to be independent from itself.
Any two substitutions $s_1,s_2 \in S$ with $s_1 = s_2$ are parallel independent if  
\begin{align}
 o_{1}(ST_1) \cap o_{2}(ST_2) \subseteq o_{1}(l_1(CP_1)) \cap o_{2}(l_2(CP_2)) \label{s1s2target2}
 \end{align}
 holds true \cite{FAGT}.

When considering only transitions, since $CP_1$ and $CP_2$ only contain places and since $s_1 = s_2$ $ST_1$ and $ST_2$ contain the same substitution transition, it follows:

 \begin{align}
(o_{1T}(ST_1) \cap o_{2T}(ST_2) \neq \emptyset) \not \subseteq \emptyset = o_{1T}(l_1(CP_1)) \cap o_{2T}(l_2(CP_2)) \label{empty3} 
\end{align}
Thus equation \ref{s1s2target2} cannot hold true and any $s$ is not independent from itself and thus can only be applied once.
\end{proof}

\section{Related Work}
\label{s.relwork}

Besides hierarchical Petri nets based on transition substitution, nets based on place substitution and Object-Oriented Petri nets (\textit{OOPN}) were considered.
There are a number of tools similar to ReConNet. \textit{Snoopy} \cite{heiner2012} is a unifying Petri net framework with a graphical user interface. It allows the modeling and simulation of colored and uncolored Petri nets of different classes, supports analytic tools and the hierarchical structuring of models.

\textit{CPN tools} \cite{ratzer2003} is another tool for the modeling and simulation of colored Petri nets. Using a graphic user interface CPN tools features syntax checking, code generation and state space analysis.

 The \textit{HiPS} tool \cite{HiPS2017}  developed at the Department of Computer Science and Engineering, Shinshu University is a tool written in C\# and also employs a graphical user interface. HiPS is a platform for design and simulation of hierarchical Petri nets. It also provides functions of static and dynamic net analysis.

While all of these tools support the design of hierarchical Petri nets each lacks ReConNet's core feature the aspect of reconfigurability.

A use case for hierarchical Petri nets can be found in \cite{sun2014}. There hierarchical colored Petri nets are used to model the French railway interlocking system \textit{RIS} for  formal verification and logic evaluation. The RIS system is responsible for the safe routing of trains. Detailed verifications and evaluations are mandatory before deploying an RIS, since it is a safety critical system. The paper describes how the signaling control and the railway road layout are specified and constructed into a colored hierarchical Petri net. 

 \cite{zhang2009} uses hierarchical colored Petri nets to model the production process of a cold rolled steel mill. For this a crude description of the entire running process of the system is given at the main net, and the more detailed behaviors are specified in the subnets. It is shown that the design is highly consistent with real production, improving the development efficiency for production planning and scheduling.

\section{Conclusion}
\label{s.conc}
This paper provides  the basics of substitution transitions for \HRPN s. 
The main contribution is there a formal definition of the \HRPN s and its flattening construction.

This work presents a step to the  integration of reconfigurable hierarchical  Petri nets into the \textit{ReConNet} tool \cite{PEOH12,reconnet}.
Ongoing work will  accomplish support of hierarchical Petri nets in \textit{ReConNet}. 
First hierarchy needs to  be introduced into \textit{ReConNet} to allow transformation simulation, including an appropriate update to \textit{ReConNet}'s  persistence module to allow proper storing and restoring of hierarchical nets. Then individual rules are added to allow the functionality of a reconfigurable net. For net verification and validation purposes the flat representation of the \HRPN \ will be used.

\bibliography{subTMadTr} 
\bibliographystyle{alpha}

\newpage
\begin{appendix}
\section{Review of Decorated Place/Transition Nets}
\label{s.decoPT}

Let us revisit the algebraic notion of Petri nets. A marked place/transition net is given by
 $N=(P,T,pre,post,M)$ with pre and post domain functions $pre,post: T \to P^\oplus$ and a current  marking $M \in P^\oplus$, 
 where $P^\oplus$ is the free commutative monoid over the set $P$ of places.   
 For $M_1, M_2 \in P^\oplus$ we have $M_1 \leq M_2$ if $M_1(p) \leq M_2(p)$ for all $p \in P$.  A transition $t \in T$ is $M$-enabled for a marking
$M \in P^\oplus$ if we have $pre(t) \leq M$, and in this case the follower
marking $M^\prime$ is given by $M^\prime=M \ominus pre(t) \oplus post(t)$ and $M\fire{t} M^\prime$ is called firing step. 
Parallel firing of an firing vector $M\fire{v} M^\prime$ can be computed using the pre and post domain functions $M^\prime= M - pre^\oplus(v) +post^\oplus(v)$.\\
The transition labels may change when the transition fires.
This  feature has been introduced in \cite{Pad12} and most of the following section is from that paper. This feature  is important for the application of a rule after a transition has already fired and cannot be modelled  without changing the labels. Considering the tokens in the post place of the transition
does not work, because these tokens may be consumed as well. 
 The extension to changing labels is 
  conservative with respect to Petri nets as it does not alter the net's behaviour, but it is crucial for the control of rule application and transition firing.
  
Morphisms of decorated place/transition  nets are given as a pair of mappings for the places and the transitions, so that the structure and the decoration is preserved and the marking may be mapped strict, yielding an $\M$-adhesive category (see Lemma~1 in \cite{Pad12}). 
\begin{definition}[Decorated place/transition  net]   
A decorated place/transition  net  is a marked P/T net $N=(P,T,pre,post,M)$ together with
			\begin{itemize}
				\item a capacity as a function $cap :P\to \Nat$ 
				\item $A_P$, $A_T$  name spaces with 
				      $\pl:P \to A_P$ and $\tl:T \to A_T$
			  \item the function $tlb: T \to W$ mapping transitions to transition labels $W$ and
			  \item the function
		          $rnw: T \to END$ where $END$ is a set containing some endomorphisms on
		          $W$, so that 
		             $rnw(t): W \to W$ is the function that renews the transition label.
			\end{itemize}
\end{definition}
The firing of these nets is the usual for  place/transition nets except for changing the transition labels.
 Moreover, this extension works for parallel firing as well.
\begin{definition}[Changing Labels by Parallel Firing]
   Given a transitions vector $v= \sum_{t \in T} k_t \cdot t$
   then the label is renewed by firing  $tlb \fire{v} tlb'$ and for each $t \in T$ the transition label 
   $tlb':T \to W$ is defined by:
    $$tlb'(t) = rnw(t)^{k_t} \circ tlb (t)$$
\end{definition}
 
In order to define rules and transformations for decorated place/transition  nets we 
introduce morphisms that map transitions to transitions  by $f_T$  and places to places by $f_P$. 
The later is extended to linear sums by $f_P^\oplus$. 
These morphisms preserve firing steps by Condition (\ref{i}) and all annotations by Condition (\ref{ii}-\ref{iv}) below. 
Since Condition (\ref{iv}) preserves the transition labels, these labels only can be changed by firing the corresponding transition, but not by transformations. Additionally, these morphisms require that the  marking at corresponding places is not decreased (Condition (\ref{v})). For strict morphisms, in addition injectivity and the preservation of markings is required (Condition (\ref{vi})).

\begin{definition}[Morphisms between  decorated place/transition nets]
Given two decorated place/ transition  nets $N_i =( P_i,T_i, pre_i, post_i,M_i,cap_i, pname_i,tname_i, tlb_i, rnw_i)$ for $i=1,2$
then 

$f:N_1 \to N_2$ is given by $f=(f_P:P_1 \to P_2,f_T:T_1 \to T_2)$ and the following equations hold: 
\begin{enumerate}
	\item \label{i}   $pre_2 \circ f_T = f_P^\oplus \circ pre_1$ and $post_2 \circ f_T = f_P^\oplus \circ post_1$ 
	\item \label{ii}  $cap_1 = cap_2 \circ f_p$ 
	\item \label{iii} $pname_1 = pname_2 \circ f_P$ 
	\item \label{iv}  $tname_1 = tname_2 \circ f_T$ and  $tlb_1 = tlb_2 \circ f_T$ and $rnw_1 = rnw_2 \circ f_T$
	\item \label{v} $M_1(p) \le M_2(f_P(p))$ for all $p \in P_1$
\end{enumerate}
Moreover, the morphism $f$ is called strict 
\begin{enumerate}
   \setcounter{enumi}{5}
	\item \label{vi} if  both  $f_P$ and $f_T$  are injective and 
 $ M_1(p) =M_2(f_P(p))$  holds for all $p \in P_1$.
\end{enumerate}
 Decorated place/transition nets together with the above morphisms yield the category $\cdecoPT$.
\end{definition}

\section{Review of $\M$-Adhesive Transformation Systems}
\label{s.madTS}
This section can be found in \cite{Padberg15} as well.

The theory of $\M$-adhesive transformation systems\footnote{See page 2 in \cite{EGHLO14} for the relation to other types of HLR systems.} has been developed as an abstract framework for 
different types of graph and Petri net 
transformation systems \cite{FAGT,EGH10}. They have been instantiated with various 
graphs, e.g., hypergraphs, attributed and typed graphs, but also with  structures, algebraic specifications and 
various Petri net classes, as  elementary nets, place/transition nets, Colored Petri nets, 
or algebraic high-level nets \cite{FAGT}.
The fundamental construct for $\M$-adhesive categories and 
systems are $\M$-van Kampen squares \cite{LS05,EGH10}\footnote{For a discussion of the various adhesive categories see page 6 in \cite{EGHLO14}.}

\begin{definition}[\M-Van Kampen square]\label{d.VK}
A pushout (\ref{eq.madhPO})
with $m \in \M$  is an \M-van Kampen square, if for 
any commutative cube (\ref{eq.madhVK}) with (\ref{eq.madhPO}) in 
the bottom \\and  
the back faces being pullbacks, the following holds:\\
 the top is pushout $\Leftrightarrow$ the
front faces are pullbacks.\\[-15mm]
\begin{minipage}[t]{0.25\textwidth}
				\begin{equation}
	         \label{eq.madhPO}
  $$
   \xymatrix@=5mm{
   			&\\
                       A  \ar[rr]|{m\in\M} \ar[dd]|{f}    
	         && B                   \ar[dd]|{g}     \\ 
	            &  \\
		       C  \ar[rr]|{n}
		    && D
		 }
     $$
		\end{equation}
		\end{minipage} 		
		\hfill~ \\[-45mm] \hspace*{7cm}
    \begin{minipage}[t]{0.45\textwidth}
				\begin{equation}
	         \label{eq.madhVK}
  $$
    \xymatrix@=3mm{
                     &&&    A'	 \ar[ddd]|<<<<<{a}    \ar[dlll]|{f'}    \ar[drr]|{m'}     
		         && \\
		                C'      \ar[ddd]|{c}                    \ar[drr]|{n'}
		   &&&&& B'      \ar[ddd]|{b}     \ar[dlll]|<<<<<{g'}          
		           &&                                                                \\
		          && D'      \ar[ddd]|<<<<<{d}                                        \\
                        &&& A                       \ar[dlll]|>>>>>{f}     \ar[drr]|{m}           \\
		                C                                      \ar[drr]|{n}
		   &&&&& B                       \ar[dlll]|{g}                       
		           &&                                                              \\
		          && D
		   } 	        	 
  $$
		\end{equation}
		\end{minipage} 
\end{definition}
\M-adhesive transformation systems can be seen as an abstract 
 transformation systems in the double pushout approach based on \M-adhesive  categories \cite{EGH10}.
\begin{definition}[$\M$-Adhesive  Category]
\label{d.madHLR}
A class $\M$ of monomorphisms in $\cC$ is called PO-PB compatible, if
\begin{enumerate}
	\item Pushouts along \M-morphisms exist and \M \ is stable under pushouts.
  \item Pullbacks along \M-morphisms exist and \M \ is stable under pullbacks.
  \item \M \ contains all identities and is closed under composition.
\end{enumerate}
 Given a PO-PB
compatible class $\M$ of monomorphisms in $\cC$, then $(\cC,\M)$ is called
$\M$-adhesive category, if pushouts along \M-morphisms are \M-Van Kampen  squares (see Def.~\ref{d.VK}).
 An \M-adhesive 
transformation system $AHS = (\cC,\M,P)$  consists of an $\M$-adhesive category $(\cC,\M)$  and a set of rules $P$. 
\end{definition}

\begin{remark}
The following kinds of Petri nets yield $\M$-adhesive categories:
\begin{itemize}
	\item PT nets and morphisms as given in Sect.~\ref{s.recPN} yield an $\M$-adhesive  category $\cPT$ (see \cite{FAGT}).  
\item Algebraic high-level nets as given in Sect.~\ref{s.recPN} have been shown in  \cite{Prange08} to be an $\M$-adhesive  category $\cAHL$ for $\M$ being the class of strict morphisms\footnote{In \cite{Prange08} they are called AHL-systems with  morphisms
that are isomorphisms on  the algebra part.}.
\item In \cite{Pad12} it is shown that
decorated place/transition nets yield an $\M$-adhesive  transformation category  $\cdecoPT$ for $\M$ being the corresponding class of strict morphisms. 
\end{itemize}

\end{remark}

\section{Additional requirements}
\label{s.addProp}

\begin{remark}
 To obtain the results for nets with subtyping  the following additional properties for the class $\M$-morphism. see\cite{FAGT}:
\begin{itemize}
	\item $\E'$-$\M'$pair factorization with $\M$-$\M'$PO-PB decomposition
 \item Initial pushouts over $\M'$-morphisms
 \item Coproducts compatible with $\M$
\end{itemize}
\end{remark}

%

\end{appendix}
\end{document}